\documentclass[12pt]{article}

\usepackage{amsmath,amssymb,amsthm}
\usepackage{mathrsfs}
\usepackage{stmaryrd}
\usepackage{geometry}
\usepackage[hidelinks]{hyperref}

\usepackage{authblk}

\title{\bf Grothendieck Topologies and Sheaf-Theoretic Foundations of
Cryptographic Security: \\
Attacker Models and $\Sigma$-Protocols as the First Step}

\author{\Large Takao Inou\'{e}}

\affil{\large Faculty of Informatics, Yamato University, \\ Osaka, Japan\footnote{Email: inoue.takao@yamato-u.ac.jp; \\ Personal Email: takaoapple@gmail.com \\ [I prefer my personal email address for correspondence.]}} 

\date{February 19, 2026}

\newtheorem{theorem}{Theorem}[section]
\newtheorem{proposition}{Proposition}[section]
\newtheorem{definition}{Definition}[section]

\newtheorem{corollary}{Corollary}[section]
\newtheorem{remark}{Remark}[section]

\begin{document}
\maketitle

\begin{abstract}
Cryptographic security is traditionally formulated using game-based
or simulation-based definitions.
In this paper, we propose a structural reformulation of cryptographic
security based on Grothendieck topologies and sheaf theory.

Our key idea is to model attacker observations as a Grothendieck site,
where covering families represent admissible decompositions of partial
information determined by efficient simulation.
Within this framework, protocol transcripts naturally form sheaves,
and security properties arise as geometric conditions.

As a first step, we focus on $\Sigma$-protocols.
We show that the transcript structure of any $\Sigma$-protocol defines
a torsor in the associated topos of sheaves.
Local triviality of this torsor corresponds to zero-knowledge,
while the absence of global sections reflects soundness.
A concrete analysis of the Schnorr $\Sigma$-protocol is provided
to illustrate the construction.

This sheaf-theoretic perspective offers a conceptual explanation
of simulation-based security and suggests a geometric foundation
for further cryptographic abstractions.
\medskip

\noindent Keywords:  Cryptographic security,
zero-knowledge proofs,
$\Sigma$-protocols,
Grothendieck topology,
sheaf theory,
attacker models,
simulation-based security,
torsors.
\medskip

\noindent MSC2020: 94A60, 18F20, 11T71.
\end{abstract}

\tableofcontents

\section{Introduction}

Cryptographic security is commonly defined using game-based or
simulation-based frameworks.
While these approaches are operationally effective,
they often obscure the structural reasons underlying security guarantees.

In this paper, we propose a geometric reformulation of cryptographic
security based on Grothendieck topologies and sheaf theory.
The central idea is to interpret attacker observations as forming a site,
where coverings encode admissible decompositions of partial information
determined by efficient simulation.

Within this framework, protocol transcripts assemble into sheaves,
and security properties emerge as intrinsic geometric features.
As a first step toward a general theory,
we focus on $\Sigma$-protocols and show that their transcript structures
naturally define torsors in an associated topos.

The present paper focuses on the core structural correspondence 
between $\Sigma$-protocols and Grothendieck-topological models; further 
technical refinements and extensions are left for future work.

\section{Related Work}

\subsection{Simulation-Based and Zero-Knowledge Security}

Zero-knowledge proofs were introduced by Goldwasser, Micali, and Rackoff
as an interactive proof system in which the verifier learns nothing
beyond the validity of the statement \cite{GMR}.
Subsequent work established simulation-based security as the standard
formalization of zero-knowledge, emphasizing the existence of efficient
simulators rather than explicit structural properties.

$\Sigma$-protocols were introduced as a fundamental class of
three-move interactive proof systems with special soundness and
honest-verifier zero-knowledge \cite{Damgard}.
They play a central role in cryptographic constructions, including
identification schemes and proof systems for discrete logarithm
relations \cite{Schnorr}.

While simulation-based definitions are operationally effective,
they remain external to the protocol structure itself.
The present work differs in that simulation is not treated as an
auxiliary algorithm but is instead encoded geometrically via
Grothendieck topologies on attacker observations.

\subsection{Game-Based Security and Attacker Models}

Game-based definitions of security have become a dominant paradigm
in modern cryptography, particularly in the analysis of complex
protocols and compositions \cite{BellareRogaway}.
These approaches model security as the inability of an efficient
adversary to distinguish between games.

Although powerful, game-based models often abstract away the internal
structure of information flow.
In contrast, the sheaf-theoretic approach adopted here makes
information restriction and recombination explicit by modeling
attacker views as objects in a site.
Covering families encode admissible decompositions of information,
rather than adversarial strategies.

\subsection{Categorical and Topos-Theoretic Approaches to Security}

Category-theoretic methods have been applied to cryptography in several
contexts, including compositional security frameworks and protocol
semantics \cite{Mitchell,AbadiGordon}.
More recently, topos-theoretic and logical methods have been explored
in the semantics of computation and information flow \cite{AbramskyCoecke}.

Sheaf-theoretic techniques have been used extensively in logic,
geometry, and distributed systems to formalize locality and
context-dependence.
However, their application to cryptographic security remains limited.

The present work is distinguished from existing categorical approaches
by its explicit use of Grothendieck topologies to model attacker
observations.
Rather than focusing on compositional semantics or process calculi,
we interpret cryptographic security properties as intrinsic geometric
features of sheaves and torsors in an associated topos.

\subsection{Position of the Present Work}

This paper does not aim to replace existing cryptographic security
definitions.
Instead, it provides a conceptual and geometric reinterpretation of
simulation-based security for $\Sigma$-protocols.

By identifying zero-knowledge with local triviality and soundness with
the absence of global sections, we offer a structural explanation of
well-known cryptographic phenomena.
This perspective suggests a unifying framework in which different
security notions may be compared and generalized using geometric tools.

\section{Attacker Models as Grothendieck Sites}

We model attacker observations by a category whose objects represent
partial views of protocol executions and whose morphisms correspond
to restriction of information.
This category captures the ways in which an attacker may lose or forget
information.

A Grothendieck topology is imposed by declaring a family of morphisms
to be covering if the information in the target object can be efficiently
simulated from the source objects.
We refer to this topology as the \emph{attacker topology}.

\section{The Schnorr $\Sigma$-Protocol as an Attacker Sheaf}

We illustrate the framework using the Schnorr $\Sigma$-protocol.
Attacker views are modeled as partial transcripts
\[
a,\quad (a,e),\quad (a,e,z),
\]
with morphisms given by erasure of transcript components.

Covering families correspond to simulatable decompositions of views.
The associated transcript presheaf assigns to each view the set of
compatible internal randomness values.
This presheaf admits a natural action by re-randomization of the prover's
randomness, yielding a torsor structure.

This concrete example serves as a guiding instance for the general
theory developed below.

\section{Example: The Schnorr $\Sigma$-Protocol as a Sheaf over Attacker Views}

We illustrate our framework by modeling the Schnorr $\Sigma$-protocol
as a sheaf over a Grothendieck site encoding attacker observations.

\begin{definition}[Attacker Observation Category]
Let $G$ be a cyclic group of prime order $q$ with generator $g$.
Let $y = g^x$ be a public value, where $x \in \mathbb{Z}_q$ is the prover's secret.

We define a category $\mathcal{C}_{\mathrm{att}}$ as follows.
\begin{itemize}
  \item Objects of $\mathcal{C}_{\mathrm{att}}$ are \emph{attacker views}, represented by
  partial transcripts of the Schnorr protocol, such as
  \[
  a = g^r,\quad (a,e),\quad (a,e,z),
  \]
  where $e \in \mathbb{Z}_q$ is a challenge and $z = r + ex$ is a response.
  \item Morphisms are restriction maps corresponding to forgetting components
  of transcripts, e.g.
  \[
  (a,e,z) \to (a,e), \qquad (a,e) \to a.
  \]
\end{itemize}
This category is a poset-like category encoding information loss
under attacker observation.
\end{definition}

\begin{definition}[Attacker Grothendieck Topology]
A family of morphisms $\{U_i \to U\}$ in $\mathcal{C}_{\mathrm{att}}$
is declared a covering family if the attacker can computationally
simulate the joint distribution on $U$ from the data available on the $U_i$.

This Grothendieck topology is called the \emph{attacker topology}.
\end{definition}

Covering families are understood up to computational indistinguishability,
rather than exact reconstruction.

Eexample. 
In the Schnorr protocol, the family
\[
\{ (a,e) \to a,\; (a,z) \to a \}
\]
is a covering family of the object $a$, since the distribution of commitments
can be simulated without knowledge of the secret $x$.

\begin{definition}[Transcript Presheaf]
Define a presheaf
\[
\mathcal{F} : \mathcal{C}_{\mathrm{att}}^{op} \to \mathbf{Set}
\]
by assigning to each attacker view the set of transcripts consistent
with that observation.
Explicitly,
\begin{align*}
\mathcal{F}(a) &= \{ r \in \mathbb{Z}_q \mid a = g^r \},\\
\mathcal{F}(a,e) &= \{ (r,z) \in \mathbb{Z}_q^2 \mid z = r + ex \},\\
\mathcal{F}(a,e,z) &= 
\begin{cases}
\{\ast\}, & \text{if } g^z = a \cdot y^e,\\
\varnothing, & \text{otherwise}.
\end{cases}
\end{align*}
Restriction maps are induced by forgetting transcript components.
\end{definition}

\begin{proposition}
The presheaf $\mathcal{F}$ satisfies the sheaf condition
with respect to the attacker topology.
\end{proposition}

\begin{proof}[Proof sketch]
Given compatible local transcript data on a covering family,
a global transcript can be reconstructed up to computational
indistinguishability by the standard Schnorr simulator.
This corresponds precisely to the gluing condition for sheaves.
\end{proof}

\begin{proposition}
The sheaf $\mathcal{F}$ is a torsor under the additive action of
$\mathbb{Z}_q$ induced by re-randomization of the prover's nonce.
\end{proposition}

\begin{remark}
The torsor $\mathcal{F}$ is locally trivial with respect to the
attacker topology, but admits no global section selecting the secret $x$.
Thus,
\begin{itemize}
  \item zero-knowledge corresponds to local triviality,
  \item knowledge soundness corresponds to the absence of global sections.
\end{itemize}
\end{remark}

\section{General $\Sigma$-Protocols as Sheaves and Torsors}

We now generalize the previous example to arbitrary $\Sigma$-protocols
and show that their security properties admit a uniform sheaf-theoretic formulation.
\medskip

\noindent \bf Remark\rm . The acting group is induced by re-randomization of the prover's internal randomness.

\begin{definition}[$\Sigma$-Protocol]
A \emph{$\Sigma$-protocol} consists of three polynomial-time algorithms
\[
(\mathsf{Commit}, \mathsf{Challenge}, \mathsf{Response})
\]
satisfying completeness, special soundness, and honest-verifier zero-knowledge.
The protocol produces transcripts of the form $(a,e,z)$.
\end{definition}

\begin{definition}[Category of Attacker Views]
Let $\Pi$ be a $\Sigma$-protocol.
We define the \emph{attacker observation category} $\mathcal{C}_{\Pi}$ as follows.
\begin{itemize}
  \item Objects are partial transcripts obtained by hiding some components
  of $(a,e,z)$.
  \item Morphisms are restriction maps corresponding to erasing information.
\end{itemize}
Thus $\mathcal{C}_{\Pi}$ encodes all possible observational contexts
available to an attacker.
\end{definition}

\begin{definition}[Attacker Topology]
A family of morphisms $\{U_i \to U\}$ in $\mathcal{C}_{\Pi}$
is a covering family if there exists a probabilistic polynomial-time simulator
that can reconstruct the distribution on $U$ from the data on the $U_i$.

This Grothendieck topology is called the \emph{attacker topology}
associated with $\Pi$.
\end{definition}

\begin{definition}[Transcript Presheaf]
Define a presheaf
\[
\mathcal{F}_{\Pi} : \mathcal{C}_{\Pi}^{op} \to \mathbf{Set}
\]
by assigning to each attacker view $U$ the set of transcripts
consistent with that view.
\end{definition}

\begin{theorem}[Zero-Knowledge as the Sheaf Condition]
Let $\Pi$ be a $\Sigma$-protocol.
If $\Pi$ satisfies honest-verifier zero-knowledge,
then the presheaf $\mathcal{F}_{\Pi}$ is a sheaf
with respect to the attacker topology.
\end{theorem}

\begin{proof}[Proof sketch]
Honest-verifier zero-knowledge provides a simulator that produces
indistinguishable transcripts locally on each attacker view.
Compatibility of local data corresponds to consistency of simulated transcripts,
and the existence of a global transcript up to indistinguishability
corresponds to the gluing axiom of sheaves.
\end{proof}

\begin{definition}[Re-randomization Group]
Let $G_{\Pi}$ denote the group acting on transcripts
by re-randomization of the prover's internal randomness.
\end{definition}

\begin{theorem}[$\Sigma$-Protocols as Torsors]
Let $\Pi$ be a $\Sigma$-protocol with honest-verifier zero-knowledge.
Then the sheaf $\mathcal{F}_{\Pi}$ is a torsor under the action of $G_{\Pi}$
in the topos of sheaves on $(\mathcal{C}_{\Pi}, J_{\Pi})$.
\end{theorem}

\begin{proof}[Proof sketch]
Local trivializations are induced by the simulator,
which provides canonical local sections.
Non-existence of a global section follows from soundness:
a global section would determine a witness,
contradicting the zero-knowledge property.
\end{proof}

\begin{corollary}[Structural Interpretation of Security Properties]
For a $\Sigma$-protocol $\Pi$:
\begin{itemize}
  \item honest-verifier zero-knowledge corresponds to local triviality
        of the torsor $\mathcal{F}_{\Pi}$,
  \item special soundness corresponds to the obstruction
        to the existence of global sections.
\end{itemize}
\end{corollary}

\noindent \bf Remark\rm. 
This formulation is independent of the concrete algebraic structure
of the underlying protocol and depends only on its transcript structure
and simulation properties.

\section{Discussion and Future Work}

The sheaf-theoretic formulation presented here provides a structural
explanation of simulation-based security for $\Sigma$-protocols.
Rather than treating simulation as an external artifact,
it emerges naturally from the geometry of attacker observations.

This work represents a first step toward a geometric foundation of
cryptographic security.
Extensions to malicious verifiers, protocol composition,
and stronger security notions are left for future investigation.

We further hope to investigate composability in this line deeply.

$$ $$

\noindent Takao Inou\'{e}

\noindent Faculty of Informatics

\noindent Yamato University

\noindent Katayama-cho 2-5-1, Suita, Osaka, 564-0082, Japan

\noindent inoue.takao@yamato-u.ac.jp
 
\noindent (Personal) takaoapple@gmail.com (I prefer my personal mail)


\begin{thebibliography}{99}

\bibitem{AbadiGordon}
M.~Abadi and A.~D.~Gordon,
\newblock \emph{A Calculus for Cryptographic Protocols: The Spi Calculus},
\newblock Information and Computation, 148(1):1--70, 1999.

\bibitem{AbramskyCoecke}
S.~Abramsky and B.~Coecke,
\newblock \emph{Categorical Quantum Mechanics},
\newblock In \emph{Handbook of Quantum Logic and Quantum Structures},
Elsevier, 2009.

\bibitem{BellareRogaway}
M.~Bellare and P.~Rogaway,
\newblock \emph{Introduction to Modern Cryptography},
\newblock UCSD Course Notes, 2005.

\bibitem{CramerDamgardSchoenmakers1994}
R.~Cramer, I.~Damg{\aa}rd, and B.~Schoenmakers,
\newblock Proofs of partial knowledge and simplified design of witness hiding protocols,
\newblock In \emph{Advances in Cryptology -- CRYPTO '94},
Lecture Notes in Computer Science, vol.~839,
Springer, 1994, pp.~174--187.

\bibitem{Damgard}
I.~Damg{\aa}rd,
\newblock \emph{On $\Sigma$-Protocols},
\newblock Lecture Notes, Aarhus University, 2002.

\bibitem{GMR}
S.~Goldwasser, S.~Micali, and C.~Rackoff,
\newblock \emph{The Knowledge Complexity of Interactive Proof Systems},
\newblock SIAM Journal on Computing, 18(1):186--208, 1989.

\bibitem{GoldreichFOC}
O.~Goldreich,
\newblock \emph{Foundations of Cryptography, Volume~1: Basic Tools},
\newblock Cambridge University Press, 2001.

\bibitem{GoldreichFOC2}
O.~Goldreich,
\newblock \emph{Foundations of Cryptography, Volume~2: Basic Applications},
\newblock Cambridge University Press, 2004.

\bibitem{Mitchell}
J.~C.~Mitchell,
\newblock \emph{Foundations for Programming Languages},
\newblock MIT Press, 1996.

\bibitem{Schnorr}
C.~P.~Schnorr,
\newblock \emph{Efficient Identification and Signatures for Smart Cards},
\newblock In \emph{Advances in Cryptology -- CRYPTO '89},
Lecture Notes in Computer Science, vol.~435,
Springer, 1990.

\end{thebibliography}
\end{document}